\documentclass[12pt,final,
thmsa,epsf,a4paper]%{article}
{article}
\usepackage{booktabs}

\newtheorem{theorem}{Theorem}%[chapter]

\newtheorem{claim}{Claim}
\newtheorem{corollary}{Corollary}%[chapter]
\newtheorem{definition}{Definition}%[chapter]
\newtheorem{example}{Example}%[chapter]
\newtheorem{remark}{Remark}%[chapter]
\newenvironment{proof}[1][Proof]{\emph{#1.} }{\  \hfill $\square $ \vspace{5 pt}}

\usepackage{stackrel}
\usepackage{natbib}
\usepackage{amssymb}
\usepackage[dvips]{graphics}
\usepackage{amsmath}

\usepackage{mathpazo}

\usepackage{enumerate}

\usepackage{amsfonts}

\usepackage{amssymb}

\usepackage{enumerate}

\usepackage{mathrsfs}

\usepackage[usenames]{color}

\usepackage{pgf,tikz}

\usetikzlibrary{decorations.pathreplacing,angles,quotes}
\usetikzlibrary{arrows.meta}
\usetikzlibrary{arrows, decorations.markings}

\usepackage{hyperref}
\hypersetup{
    colorlinks,
    citecolor=blue,
    filecolor=black,
    linkcolor=blue,
    urlcolor=blue
}
\usepackage{caption}
\usepackage{soul}
\usepackage{pgf,tikz}
\usetikzlibrary{arrows}

\usepackage[utf8]{inputenc}
\usepackage{amssymb, amsmath,geometry}
\usepackage{fancyhdr}
\usepackage{soul}
\usepackage{cancel}

\usepackage{emptypage}

\newcommand*\samethanks[1][\value{footnote}]{\footnotemark[#1]}

\begin{document}

\title{A decomposition from a many-to-one matching market with path-independent choice functions to a one-to-one matching market\thanks{We thank anonymous referees for their very detailed comments. We also thank Marcelo Fernandez for his participation in the early stages of the paper.
We acknowledge the financial support of UNSL through grants 03-2016 and 03-1323, and from Consejo Nacional
de Investigaciones Cient\'{\i}ficas y T\'{e}cnicas (CONICET) through grant
PIP 112-200801-00655, and from Agencia Nacional de Promoción Cient\'ifica y Tecnológica through grant PICT 2017-2355.}}

\author{Pablo Neme\thanks{Instituto de Matem\'{a}tica Aplicada San Luis (UNSL-CONICET) and Departamento de Matemática, Universidad Nacional de San
Luis, San Luis, Argentina. Emails:  \texttt{pabloneme08@gmail.com} (P. Neme) and \texttt{joviedo@unsl.edu.ar} (J. Oviedo).}  \and Jorge Oviedo\samethanks[2]}

\date{\today}

\maketitle
\begin{abstract}
For a many-to-one market where firms are endowed with path-independent choice functions, based on the Aizerman-Malishevski decomposition, we define an associated one-to-one market. Given that the usual notion of stability for a one-to-one market does not fit well for this associated one-to-one market, we introduce a new notion of stability. This notion allows us to establish an isomorphism between the set of stable matchings in the many-to-one market and the matchings in an associated one-to-one market that meet this new stability criterion. Furthermore, we present an adaptation of the well-known deferred acceptance algorithm to compute a matching that satisfies this new notion of stability for the associated one-to-one market. Finally, as a byproduct of our isomorphism, we prove an adapted version of the so-called Rural Hospital Theorem.

\bigskip



\end{abstract}
\section{Introduction}
Many-to-one markets have been extensively studied in the literature, starting with the college admissions problem and extending to the assignment of medical interns to hospitals, as well as applications in labor markets. In parallel, one-to-one matching models---such as the classical marriage problem---have provided a foundational framework for analyzing stability in simpler settings. These two types of models are closely connected, and several results in the literature show how many-to-one environments can be decomposed into equivalent one-to-one representations. Such decomposition techniques have proven useful for transferring structural properties and solution concepts between the two frameworks, and for deepening our understanding of the relationship between different matching environments.

The Aizerman-Malishevski decomposition, a classical result in the choice literature \citep[e.g.,][]{Moulin1985}, states that any path-independent choice function can be represented as the union of linear orders over individuals. This decomposition has proven especially useful in matching theory. For instance, \citet{chamb2017choice} employ it to analyze many-to-many matching markets with contracts, showing how it enables the use of deferred acceptance algorithms under path-independent choice. In our context, this decomposition allows us to reinterpret a firm with a path-independent choice function as a collection of ``copies'' of itself, each endowed with strict preferences over workers. This transformation defines an associated one-to-one market whose structure reflects the substitutability encoded in the firms' original choices. However, since the copies of a single firm may rank workers differently, the classical notion of stability in one-to-one markets does not translate adequately. In Section~\ref{seccion models}, we illustrate this issue through Example~\ref{ejemplo}, where a many-to-one market with four stable matchings corresponds to an associated one-to-one market with only one stable matching under the usual definition. To bridge this gap, we propose a new notion of stability, which we call \emph{stability*}. While it retains the key elements of classical stability---individual rationality and absence of blocking pairs---it is designed to accommodate the structure induced by the decomposition. Under this refined notion, we prove that the sets of stable matchings in the many-to-one market and stable* matchings in an associated one-to-one market are isomorphic.

To highlight the difference between the usual notion and stability*, we introduce two additional conditions not typically required in one-to-one settings. First, for individual rationality* in this associated market, we impose an envy-free condition among firm-copies: no firm-copy matched to a worker should prefer another worker matched to a different copy of the same firm. Second, in the presence of blocking pairs, we require that the worker in the blocking pair is preferred by the relevant firm-copy over all workers matched to other copies of that same firm.

A common approach to prove the non-emptiness of a stable set is constructive. In the seminal paper by \cite{Gale1962}, an algorithm is presented ---the well-known deferred acceptance algorithm--- which constructs a stable matching for traditional one-to-one markets.
Although our isomorphism demonstrates that the set of stable* one-to-one matchings is non-empty, we adapt the deferred acceptance algorithm to a related one-to-one market, taking into account the specific characteristics of stability*. Although one-to-one markets are traditionally symmetric, our related one-to-one market is not. This is because it originates from the decomposition of path-independent choice functions. Thus, we present two adapted versions of the deferred acceptance algorithm: one where firm-copies propose and another where workers propose. We show that, whether the firm-copies or the workers are the proposers, the algorithm returns a stable* matching. In this way, independently of the isomorphism, we establish that the set of stable* matchings for the related one-to-one market is non-empty.

The idea of decomposing a many-to-one market into a one-to-one market has been previously studied under more restrictive assumptions. When assuming that firms have preferences and those preferences are assumed to be \emph{responsive} to an individual ranking of workers, \cite{rothsotomayot1989} demonstrate that each firm can be decomposed into identical units (or copies) according to its capacity (quotas $q$). A key distinction in this decomposition is that each of these copies shares the same individual preferences over workers (derived from responsive preferences), and workers rank all copies of a given firm above those of any other firm, preserving the same order of preferences across all copies. This decomposition transforms a many-to-one market with responsive preferences into a corresponding one-to-one market. Furthermore, \cite{rothsotomayot1989} establish an isomorphism between the stable matchings of a many-to-one market and those of an associated one-to-one market. %This decomposition naturally facilitates extending various results from the one-to-one model to the many-to-one setting. %However, not all results from the one-to-one model can be extended to many-to-one models, even with preference structures as restrictive as the responsiveness condition.  \citep[see discussion in Section 5.3.2 ][ about ``weak Pareto optimality'']{RothSotomayor90}.

The paper is organized as follows. In Section \ref{seccion models}, we present the many-to-one market, the Aizerman-Malishevski decomposition, and the associated one-to-one market. For this associated one-to-one market, in Section \ref{seccion resultados}, we present an adapted deferred acceptance algorithm, an isomorphism with the many-to-one market, and an adapted version of the Rural Hospital Theorem. Finally, in Section \ref{seccion final remaks} some final remarks are presented.

\section{Preliminaries}\label{seccion models}
This section contains two subsections. In Subsection \ref{subseccion mercado M-1}, we present the many-to-one market where firms are endowed with path-independent choice functions, and the respective notions of matching, individual rationality, and stability. In Subsection \ref{subseccion mercado 1-1}, we introduce the Aizerman-Malishevski decomposition of a path-independent choice function into linear orders, and we show a decomposition of the many-to-one market into an associated one-to-one market with its respective notions of matching, individual rationality*, and stability*.

\subsection{Many-to-one matching market}\label{subseccion mercado M-1}

We consider a many-to-one matching market where $\boldsymbol{\Phi} =\{\varphi _1,\dots ,\varphi _n\}$
denotes the set of $n$ firms and $\boldsymbol{{\mathcal{W}}}=\{w_1,\dots ,w_k\}$ the
set of $k$ workers. Each firm $\varphi \in \Phi $  is endowed with a choice function $C_\varphi:2^{\mathcal{W}}\to 2^{\mathcal{W}}$ such that for each $W\in \mathcal{W}$, $C_\varphi(W)\subseteq W.$ The idea behind a choice function is that if $W$ is an available set of workers, $C_\varphi(W)$ si the set that firm $\varphi$ chooses. Each worker $w\in {\mathcal{W}}$ has a strict, transitive, and
complete preference relation $P_{w}$ over ${\Phi }\cup \left \{\emptyset
\right \}$.  Let $R_w$ denote the weak preference relation associated with $P_w$, meaning that if $\varphi_iR_w\varphi_j$ imply that either $\varphi_iP_w\varphi_j$ or $\varphi_i=\varphi_j$. %The collection of these preference relations forms the \textbf{preference profile} $\boldsymbol{P}=\left((P_{\varphi })_{\varphi \in \Phi },(P_w)_{w\in {{\mathcal{W}}}%}\right )$.
%Given a subset of workers $W\subseteq {{{\mathcal{W}}}}$, let $%
%C_{\varphi }(W)$ denote firm $\varphi$'s choice function
%induced by preference relation $P_{\varphi }$; 
%that is, $C_{\varphi }(W)=\{W'\subseteq {{\mathcal{W}}}:W'=\max_{W''\subseteq W} P_{\varphi }\} $. 
Firm $\varphi $'s choice function is \textbf{path-independent} if for $W,W'\in \mathcal{W}$ we have $C_\varphi(W\cup W')=C_\varphi(C_\varphi(W)\cup W').$\footnote{\cite{aizerman1981general} establishes that if a choice function is path-independent if and only if it satisfies \textbf{substitutability} and \textbf{consistency}. Substitutability states that if for each $ W\subseteq \mathcal{W}$ we have that $w\in C_{\varphi }(W)$ implies that $w\in C_{\varphi }(W\setminus \{w'\})$  for each $w'\in W\setminus \{w\}$. Consistency states that for $W,W'\in \mathcal{W},$ if $C_\varphi(W)\subseteq W'\subseteq W$ then $C_\varphi(W)=C_\varphi(W')$.  } Throughout the paper, we assume that firms in the many-to-one matching market have path-independent choice functions. We denote the \textbf{many-to-one matching market} as $\boldsymbol{\mathcal{M}}=(\Phi,\mathcal{W},C,P)$ where $\Phi$ is the set of firms, $\mathcal{W}$ is the set of workers, $C$ is the profile of choice functions for firms, and $P$ is the profile of preferences of workers. 

%\footnote{The law of aggregate demand (LAD) states: if $X\subseteq X^{\prime }$, then $|C(X)|\leq |C(X^{\prime })|$. LAD has not been used in establishing any of the results described in 1.-5.(a); however, it could prove useful in establishing 5.(b), since under substitutable preferences and LAD one regains a lot of structure in the problem; e.g. \cite{Blair1988}'s order is the joint operation associated with the lattice structure of the stable set, the rural hospital theorem holds, etc.}

A many-to-one \textbf{matching} is a mapping $\mu : \Phi \cup {{\mathcal{W}}} \rightarrow 2^{\Phi \cup {{\mathcal{W}}}}$ that satisfies the following conditions: (i) $\mu(w) \in \Phi \cup {\emptyset}$ for each $w \in {{\mathcal{W}}}$, (ii) $\mu(\varphi) \in 2^{{{\mathcal{W}}}}$ for each $\varphi \in \Phi$, and (iii) $\mu(w) = \varphi$ if and only if $w \in \mu(\varphi)$. A matching $\mu$ is \textbf{blocked} by a worker $w$ if $\emptyset  P_{w}  \mu(w)$. Similarly, $\mu$ is blocked by a firm $\varphi$ if $\mu(\varphi) \neq C_{\varphi}(\mu(\varphi))$. We say that a matching is \textbf{individually rational} if it is not blocked by any individual agent. A matching $\mu$ is blocked by a worker-firm pair $(w, \varphi)$ if $w \notin \mu(\varphi)$, $w \in C_{\varphi}(\mu(\varphi) \cup \{w\})$, and $\varphi P_{w} \mu(w)$. Finally, we say that a matching $\mu$ is \textbf{stable} if it is not blocked by any individual agent or by any worker-firm pair.
We denote by $\boldsymbol{S({{\mathcal{M}}})}$ to the \textbf{set of all stable matchings} for the
matching market ${{\mathcal{M}}}$.

 \cite{martinez2008invariance} show that the set of stable matchings remains the same when an agent’s preference relation changes, as long as these two preference relations have the same induced choice function. Since it is well known that when preferences are substitutable, the set of stable matchings is non-empty, then, in our context, when firms have path-independent choice functions, the set of stable matchings is non-empty \cite[see][for more details]{chamb2017choice}.

\subsection{An associated one-to-one market}\label{subseccion mercado 1-1}

In this subsection, we present the \emph{Aizerman-Malishevski decomposition} \citep[see, for instance,][]{aizerman1981general,chamb2017choice}. This decomposition allows us to establish a connection between a many-to-one market with path-independent choice functions and a one-to-one market. It captures the essence of the path-independence (and, therefore, substitutability) on firms' choice functions by splitting each firm into multiple copies and making that each copy of a firm behaves independently of the others, meaning that the copies have different linear preferences over workers. 

\begin{description}

\item[Aizermann-Malishevski decomposition:] 
Choice function $C_{\varphi }$ is path-independent if and only if there is a finite sequence of preference relations 
$\{P_{\varphi j}\}_{j\in J_{\varphi }}$ on ${{\mathcal{W}}}$ 
such that, for all $S\subseteq {{\mathcal{W}}}$: 
\begin{equation*}
C_{\varphi }(S)=\bigcup_{j\in J_{\varphi }}\max_{S}P_{\varphi j}.
\end{equation*}

\end{description}

Given a firm $\varphi_i$ and its path-independent choice function $C_{\varphi_i}$, the decomposition is carried out as follows. Take $S\subseteq \mathcal{W}$ and begin with an arbitrary worker \( w_1 \in C_{\varphi_i}(S) \). Next, we choose a worker from \( C_{\varphi_i}(S \setminus \{w_1\}) \), say \( w_2 \). We continue iteratively by selecting \( w_k \in C_{\varphi_i}(S \setminus \{w_1, \ldots, w_{k-1}\}) \). In this way, we obtain the preference ordering \( P_{\varphi_i j} : w_1, w_2, \ldots \). Considering all possible sequences of such choices, we obtain the collection of preference relations \( \{P_{\varphi_i j}\}_{j \in J_{\varphi_i}} \), which constitute an Aizerman–Malishevski decomposition. Note that this decomposition (i.e., the indices of the linear preferences in the decomposition) depends on the order in which workers are selected.

Next, using the Aizerman–Malishevski decomposition, we introduce a \textbf{related one-to-one market}, denoted by $\boldsymbol{M}$, and provide a suitable notion of stability for this newly associated market. It is worth noting that the related one-to-one market is not unique: it depends on the order in which workers are considered when applying the Aizerman–Malishevski decomposition.

Let $${{\mathcal{F}}}=\{f_{11},\dots ,f_{1J_{1}},\dots ,f_{n1},\dots,f_{nJ_n}\}$$ be the set of firm-copies, where $f_{ij}$ denotes the firm $\varphi
_i $'s $j$-th copy. These copies are taken from the Aizermann-Malishevski
decomposition, so that $C_{\varphi _i}(S)=\bigcup _{j\in J_i}\max
_S P_{\varphi_ij}$.\footnote{Henceforth, to easy notation, we write $P_{ij}$ instead of $P_{\varphi_{i}j}.$} 
The set of workers is ${{\mathcal{W}}}=\{w_1,\dots ,w_k\}$. 
Now, each worker $w$ has a strict, transitive, and complete preference relation $\overline {P}_w$
on the set ${{\mathcal{F}}}\cup \{\emptyset \}$ such that:

\begin{enumerate}[(1)]

\item \label{pref.c.1}
$\varphi _i P_w \varphi _{i^{\prime }}$ implies that $f_{ij}\overline{P}_w f_{i^{\prime }j^{\prime }}$ for each $j\in J_i$ and each $j^{\prime }\in J_{i^{\prime }}$. \label{cond.i}  

% \textbf{Shouldn't it be $\iff$?}
\item \label{pref.c.2}
$f_{ij}\overline {P}_wf_{ij^{\prime }} \text{ if and only if }
j<j^{\prime }$, for each $i\in I=\{1,\dots ,n\}$ and each $j,j'\in J_i.$\label{cond.ii}
%\item \label{pref.c.2}
%$\left (\forall i\in I=\{1,\dots ,n\}\;\&\;\forall j\in J_i,\forall
%j^{\prime }\in J_i\right )\quad \left [f_{ij}\overline {P}_wf_{ij^{\prime }}\iff
%j<j^{\prime }\right ]$ \label{cond.ii}
\end{enumerate}

Condition \eqref {cond.i} ensures that workers' preferences ($\overline {P}_w$)
over firm-copies reflect exactly the preference relation ($P_w$%
) over firms. In turn, condition \eqref {cond.ii} ensures that $\overline {%
P}_w$ is a linear order. We denote the associated one-to-one matching market
as $M=({{\mathcal{F}}},{{\mathcal{W}}},(\overline {P}_w)_{w\in W},(P_{ij})_{i\in
I,j\in J_i})$. 

The following example illustrates how we can decompose a many-to-one market where firms have path-independent choice functions into an associated one-to-one market using the Aizermann-Malishevski decomposition. It is well known that each path-independent choice function induces a substitutable preference relation \citep[see][among others]{martinez2008invariance,chamb2017choice}. Therefore, for simplicity in the example, we assume that firms have a substitutable preference relation satisfying also consistency.
\begin{example}\label{ejemplo}
Consider a many-to-one matching market $(\Phi,\mathcal{W}, P)$, where $\Phi=\{\varphi_1,\varphi_2\}$, $\mathcal{W}=\{w_1,w_2,w_3,w_4\}$ and the preference profile is presented in the following table:\footnote{To ease notation, we assume that firms' preference relations list only acceptable subsets, omitting curly brackets. For instance, we write $P_{\varphi}: w_1w_2, w_1, w_2$ instead of $P_{\varphi}: \{w_1, w_2\}, \{w_1\}, \{w_2\}, \emptyset$. Similarly, for workers' preferences, we list only the acceptable firms.
} 
\begin{align*}
P_{\varphi_{1}}: w_1w_2, w_1w_3, w_2w_4, w_3w_4, w_1, w_2, w_3, w_4 &  & P_{w_1}=P_{w_2}:\varphi_{2}, \varphi_{1}\\
P_{\varphi_{2}}: w_3w_4, w_1w_3, w_2w_4, w_1w_2, w_4, w_3, w_2, w_1  &  & P_{w_3}=P_{w_4}: \varphi_{1}, \varphi_{2}%
\end{align*}
The associated Aizermann-Malishevski decomposition of the path-independent choice functions into linear preferences is presented in the following table:
\begin{align*}
P_{11}: w_1,w_2,w_3,w_4  & & P_{21}: w_3,w_4,w_2,w_1\\
P_{12}: w_1,w_2,w_4,w_3  & & P_{22}: w_3,w_4,w_1,w_2\\
P_{13}: w_1,w_4,w_2,w_3  & & P_{23}: w_3,w_2,w_4,w_1\\
P_{14}: w_2,w_1,w_3,w_4  & & P_{24}: w_4,w_3,w_2,w_1\\
P_{15}: w_2,w_1,w_4,w_3  & & P_{25}: w_4,w_3,w_1,w_2\\
P_{16}: w_2,w_3,w_1,w_4  & & P_{26}: w_4,w_1,w_3,w_2
\end{align*}
Now, the preferences of workers over firm-copies are presented in the following table:
\begin{align*}
 \overline{P}_{w_1}=\overline{P}_{w_2}:f_{21},f_{22},f_{23},f_{24},f_{25},f_{26},f_{11},f_{12},f_{13},f_{14},f_{15},f_{16}\\
 \overline{P}_{w_3}=\overline{P}_{w_4}: f_{11},f_{12},f_{13},f_{14},f_{15},f_{16},f_{21},f_{22},f_{23},f_{24},f_{25},f_{26}%
\end{align*}
Thus, the one-to-one associated market is denoted by $(\mathcal{F},\mathcal{W},P_{\mathcal{F}},\overline{P}_{\mathcal{W}}),$ where $$\mathcal{F}=\{f_{11},\ldots,f_{16},f_{21},\ldots,f_{26}\}$$ and $P_\mathcal{F}$ is the preference profile derived of the Aizermann-Malishevski decomposition.  \hfill$\Diamond$
\end{example}

A matching in the associated one-to-one matching market is a
mapping $\lambda :{{\mathcal{F}}}\cup {{\mathcal{W}}}\rightarrow {{\mathcal{F%
}}}\cup {{\mathcal{W}}}\cup \{\emptyset\}$ such that $\lambda (f)\in {{\mathcal{W}}}\cup
\{\emptyset \}$, $\lambda (w)\in {{\mathcal{F}}}\cup \{\emptyset \}$ and $%
\lambda (f)=w\iff \lambda (w)=f$.

Given that one of our goals in this paper is to show an isomorphism between the set of stable matchings in the many-to-one market where firms have path-independent choice functions and a one-to-one market, it is necessary to adapt the notions of individual rationality and stability for the associated one-to-one market. The following example illustrates this need:\medskip

\noindent \textbf{Example \ref{ejemplo} (continued)} \textit{The stable matchings of the many-to-one market are:}
\begin{align*}
    \mu_\Phi=
\setcounter{MaxMatrixCols}{14}
\begin{pmatrix}
 \varphi_1 & \varphi_2  \\
 w_1w_2 &w_3w_4
 \end{pmatrix}
     \mu_1=
\setcounter{MaxMatrixCols}{14}
\begin{pmatrix}
 \varphi_1 & \varphi_2  \\
 w_1w_3 &w_2w_4
 \end{pmatrix}\\
     \mu_2=
\setcounter{MaxMatrixCols}{14}
\begin{pmatrix}
 \varphi_1 & \varphi_2  \\
 w_2w_4 &w_1w_3
 \end{pmatrix}
     \mu_\mathcal{W}=
\setcounter{MaxMatrixCols}{14}
\begin{pmatrix}
 \varphi_1 & \varphi_2  \\
 w_3w_4 &w_1w_2
 \end{pmatrix}
\end{align*}
\textit{Let $\lambda$ be a matching in the related one-to-one market. Consider the standard notion of individual rationality, i.e., $\lambda(f_{ij}) R_{ij} \emptyset$ and $\lambda(w) \overline{R}_w \emptyset$, and the standard notion of stability, i.e., there is no blocking pair $(f_{ij}, w)$ such that $w \neq \lambda(f_{ij})$, $w P_{ij} \lambda(f_{ij})$, and $f_{ij} \overline{P}_w \lambda(w)$. }

\textit{We can observe that the only one-to-one matching in the related market that satisfies this notion of individual rationality and that has no blocking pairs with respect to this blocking notion is the following:}
\[
\lambda=
\setcounter{MaxMatrixCols}{14}
\begin{pmatrix}
 f_{11} & f_{12} & f_{13} & f_{14} & f_{15} & f_{16} & f_{21} & f_{22} & f_{23} & f_{24} & f_{25} & f_{26} \\
 w_3 & w_4 & \emptyset & \emptyset & \emptyset & \emptyset & w_2 & w_1 & \emptyset & \emptyset & \emptyset & \emptyset 
 \end{pmatrix}.
\]

\textit{Note that if we consider any one-to-one matching that assigns workers $w_1$ and $w_2$ with any firm-copy $f_{1,\cdot}$, and workers $w_3$ and $w_4$ with any firm-copy $f_{2,\cdot}$, it will have blocking pairs: any unmatched firm-copy $f_{1\cdot}$ with workers $w_3$ or $w_4$, and similarly, any unmatched firm-copy $f_{2\cdot}$ with workers $w_1$ or $w_2$.} 
\textit{This situation clearly implies that, with these notions of individual rationality and stability, it is not possible to define an isomorphism between a set with four elements (the set of stable matchings in the many-to-one market) and a set with only one element (the set of stable matchings under these notions).}\hfill$\Diamond$\medskip

To conclude this section, we present the necessary definitions to establish a proper notion of ``individual rationality'' and ``stability'', in order to show in the following section that the sets of ``stable matchings'' in both markets are isomorphic.

\begin{definition} \label{def.block*}
A matching $\lambda $ is \textbf{blocked* by a worker }$w$ if $\emptyset $ $%
\overline {P}_{w}$ $\lambda \left (w\right )$. Moreover, a matching $\lambda $ is \textbf{blocked* by
a firm-copy} $f_{ij}$ if either

\begin{enumerate}[(1)]

\item $\emptyset P_{ij}\lambda (f_{ij})$, or

\item $\lambda (f_{ij})\neq \emptyset $ and  there is $j^{\prime }\in
J_i$ such that $\lambda (f_{ij^{\prime }})P_{ij}\lambda (f_{ij}).$ \label{cond.b.ii}
\end{enumerate}
\end{definition}

Note that the absence of blocking* firm-copies implies the negation of condition \eqref{cond.b.ii}. So, the negation of this condition can be interpreted as an envy-free-like property among the
matched copies of a given firm, i.e., no copy that is matched wants to
switch partners with another copy. Thus, a one-to-one matching that is not blocked* by any agent (either worker or firm-copy) we say that it is \textbf{individually rational*}.

\begin{definition}\label{def de par bloqueante}
A matching $\lambda $ is said to be blocked* by a firm-copy and worker pair $(f_{ij},w)$ if

\begin{enumerate}
\item $f_{ij}\overline {P}_w\lambda (w)$; and \label{cond.b.iii}

\item $wP_{ij}\lambda
(f_{ij^{\prime }})$ for each $j^{\prime }\in J_i.$ \label{cond.b.iv}
\end{enumerate}
\end{definition}
Now, we are in a position to formally define the notion of stability* for the associated one-to-one market.
\begin{definition}
A matching $\lambda $ is stable* if it is not blocked by any individual
agent or any (firm-copy)-worker pair.
\end{definition}
Let $\boldsymbol{S^*(M)}$ denote the \textbf{set of stable* matchings} of the associated
one-to-one matching market.

Note that it is important to emphasize that condition \eqref{cond.b.iv} in Definition \ref{def de par bloqueante} is stronger than the usual notion of pairwise block, as it must account for the possibility that another copy of the same firm may be matched to a better worker. This is particularly relevant when establishing the equivalence (in an isomorphic sense) between the set of stable matchings in the many-to-one matching market and the set of stable* matchings in an associated one-to-one market since this condition partially captures the essence of the fact that firms in the original many-to-one market are endowed with path-independent choice functions.

In the following remark, we summarize the possible scenarios in which a matching in the associated one-to-one market is not stable*.

\begin{remark}\label{remark de no stable*}
If  $\lambda\notin S^{*}\left(  M\right)$, then (at least) one of the following must hold:

\begin{enumerate}[(i)]

\item There is $ w\in \mathcal{W}$ such that $\emptyset\overline{P}_{w}\lambda\left(  w\right)$.  

\item There is $ f_{ij}\in \mathcal{F}$ such that $\emptyset P_{ij}\lambda\left(  f_{ij}\right)$.

\item There are $  f_{ij},f_{ij'}\in \mathcal{F}$ such that $\lambda(f_{ij^{\prime}})P_{ij}\lambda(f_{ij})P_{ij}\emptyset$.

\item There is a pair $ (f_{ij},w)\in \mathcal{F}\times \mathcal{W}$ such that $f_{ij}\overline{P}_{w}\lambda(w)$ and $wP_{ij}\lambda(f_{ij'})$ for each $j'$.

\end{enumerate}

\end{remark}

\section{Results on the set of stable* matchings}\label{seccion resultados}
In this section, we present the main results of the paper. In Subsection \ref{subseccion algorithmos}, we adapt the well-known deferred acceptance algorithm, initially introduced in the seminal work of \cite{Gale1962}, to our one-to-one market framework and the notion of stable*. We show that in both cases, whether firms-copies propose or workers propose, the output of the adapted deferred acceptance algorithm is a stable* matching. Thus, we demonstrate that the set of stable* matchings is non-empty.

In Subsection \ref{subseccion isomorfismo}, we establish an isomorphism between the set of stable matchings in a many-to-one market where the firms have path-independent choice functions and the set of stable* matchings in an associated one-to-one market. As a result of our isomorphism, we can establish an adapted version of the so-called Rural Hospital Theorem for an associated one-to-one market.

\subsection{Adapted deferred acceptance algorithm}\label{subseccion algorithmos}

In this subsection, we introduce two variants of the deferred acceptance algorithm that compute a Stable* matching: one in which firm-copies propose and the other in which workers propose. These adaptations take into account the specific characteristics of an associated one-to-one market.

The algorithm presented in Table \ref{tabla algoritmo F-proposal} is the adapted firm-copies-proposing deferred acceptance algorithm for a one-to-one market $M$. While this algorithm follows the same principle as the one introduced by \cite{Gale1962}, our adaptation includes an important feature: a firm-copy can only make an offer to a worker if no other copy of the same firm is already matched with a more preferred worker. This restriction is essential to prevent envy between firm-copies.
Let $O^k_w$ be the set of firms making a job offer to worker $w$ at stage $k$.
\begin{table}[h!]
\centering
\begin{tabular}{l l}
\hline \hline
\multicolumn{2}{l}{\textbf{Algorithm:}}\vspace*{10 pt}\\

\textbf{Input} & A one-to-one market $M$.\\

\textbf{Output} & A matching $\lambda_{\mathcal{F}}$.\vspace*{10 pt}\\
%\textbf{Begin} & \vspace*{5 pt}\\
 & \hspace{-52pt}\texttt{DEFINE}: $\lambda^0(w)=\emptyset$ and $\lambda^0(f)=\emptyset$ for each $w\in \mathcal{W}$ and each $f\in \mathcal{F}$.\vspace*{10 pt}\\
  \textbf{Stage $\boldsymbol{1}$} &  $\boldsymbol{(a)}$ Each firm-copy $f$ makes a job offer to the best acceptable worker at $P_f$.\\
&  $\boldsymbol{(b)}$ Each worker $w$ selects, with respect to $\overline{P}_w$, the best firm-copy\\
& \hspace{20 pt} among those in $O^1_w\cup \lambda^0(w)$, say $\widetilde{f}$.\\
& \hspace{-20 pt} Set $R^1_w=O^1_w \setminus \{\widetilde{f}\}$  for each $w\in \mathcal{W}$. Define \\
& \hspace{40 pt}$\lambda^1(f_{ij})=\begin{cases}
\emptyset   & \text{if } f_{ij}\in \cup_{w\in \mathcal{W}} R_w^1, \\ 
w   & \text{if} f_{ij}=\widetilde{f}.\\
\end{cases}$ \\
 \textbf{Stage $\boldsymbol{k}$} & $\boldsymbol{(a)}$ Each firm-copy $f_{ij}\in \bigcup_{\widetilde{w}\in \mathcal{W}}R^{k-1}_{\widetilde{w}}$ is \emph{authorized} to make an offer to the best \\ 
 &\hspace{20 pt} acceptable worker $w$ with respect to $P_{f_{ij}}$ who has not previously rejected it, \\
 & \hspace{20 pt} if there are no $f_{ij'}$ and $w'\neq w$ such that $\lambda^{k-1}(w')=f_{ij'}$ and $w'P_{f_{ij}}w$.\\
& \hspace{20 pt} Otherwise, $f_{ij}$ is \emph{unauthorized} to make an offer at this stage.\\
& $\boldsymbol{(b)}$ Each worker $w$ selects, with respect to $\overline{P}_w$, the best firm-copy\\
& \hspace{20 pt} among those in $O^k_w\cup \lambda^{k-1}(w)$, say $\widetilde{f}$. \\
& \hspace{-20 pt} Set $R^k_w=O^k_w \cup \lambda^{k-1}(w) \setminus \{\widetilde{f}\}$ for each $w\in \mathcal{W}$. Define\\
& \hspace{40 pt}$\lambda^k(f_{ij})=\begin{cases}
\emptyset    & \text{if } f_{ij}\in \cup_{w\in \mathcal{W}} R_w^k, \\ 
\emptyset   & \text{if $f_{ij}$ is unauthorized to make an offer,}\\ 
w   & \text{if} f_{ij}=\widetilde{f}\text{, and}\\
\lambda^{k-1}(f_{ij})& \text{otherwise.}
\end{cases}$ \\
& \hspace{20 pt} \texttt{If} $\bigcup_{w\in \mathcal{W}}R^k_w =\emptyset$: \\
& \hspace{60 pt} \texttt{STOP} and let $\lambda_{\mathcal{F}}=\lambda^k$.\\
& \hspace{20 pt} \texttt{ELSE:} \\
& \hspace{60 pt} \texttt{CONTINUE TO STAGE $k+1$} \\
\\
\hline \hline
\end{tabular}
\caption{Adapted Firm-copies-proposal Deferred Acceptance algorithm}
\label{tabla algoritmo F-proposal}
\end{table}

\bigskip

The algorithm will eventually stop in a finite number of stages, as firms begin by making offers to the top of their preference lists and continue to make offers (if authorized) to less preferred workers as their previous offers are rejected. Given that the set of firms and workers is finite, there will eventually be no more rejections, leading to the cessation of the algorithm. 
Before presenting that the output of the algorithm is a stable* matching, we illustrate the procedure for Example \ref{ejemplo}.
\medskip

\noindent \textbf{Example \ref{ejemplo} (continued)} Considering the associated one-to-one market previously presented, we applied the adapted firm-copies-proposing deferred acceptance algorithm as follows:

\noindent\textit{\noindent\textbf{Stage 1:}  Firm-copies offers are presented in the following table:}
\begin{center}
\begin{tabular}{|c|c|c|c|c|c|c|c|c|c|c|c|}\hline
$f_{11}$&$f_{12}$&$f_{13}$&$f_{14}$&$f_{15}$&$f_{16}$&$f_{21}$&$f_{22}$&$f_{23}$&$f_{24}$&$f_{25}$&$f_{26}$\\
$w_1$&$w_1$&$w_1$&$w_2$&$w_2$&$w_2$&$w_3$&$w_3$&$w_3$&$w_4$&$w_4$&$w_4$\\ \hline
\end{tabular}
\end{center} 
So, the sets of offers that workers receive in Stage 1 are: 
\begin{align*}
    O_{w_1}^{1}=\{f_{11},f_{12},f_{13}\}\\
    O_{w_2}^{1}=\{f_{14},f_{15},f_{16}\}\\
    O_{w_3}^{1}=\{f_{21},f_{22},f_{23}\}\\
    O_{w_4}^{1}=\{f_{24},f_{25},f_{26}\}
\end{align*}
Workers, according to their preferences $\overline{P}$, accept an offer and the resulting rejection sets are: 
\begin{align*}
    R_{w_1}^{1}=\{f_{12},f_{13}\}\\R_{w_2}^{1}=\{f_{15},f_{16}\}\\R_{w_3}^{1}=\{f_{22},f_{23}\}\\R_{w_4}^{1}=\{f_{25},f_{26}\}
\end{align*}
Thus, $$
\lambda^1=
\setcounter{MaxMatrixCols}{14}
\begin{pmatrix}
 f_{11} & f_{12} & f_{13} & f_{14} & f_{15} & f_{16} & f_{21} & f_{22} & f_{23} & f_{24} &f_{25} & f_{26} \\
 w_1 &\emptyset & \emptyset & w_2 & \emptyset & \emptyset & w_3 & \emptyset & \emptyset & w_4 & \emptyset & \emptyset 
 \end{pmatrix}.
$$

\noindent\textit{\textbf{Stage 2 (and henceforth):} Each firm-copy $f_{ij} \in \bigcup_{w \in \mathcal{W}} R^{1}_{w}$ is no longer authorized to make an offer to any $w \in \mathcal{W}$. To see this consider, for instance, $f_{12} \in R^{1}_{w_1}$. We have that $w_1 = \lambda(f_{11}) P_{12} w_k$, where $k \in \{2,4\}$. Therefore, the output of the adapted firm-copies-proposing deferred acceptance algorithm is:
$$
\lambda_\mathcal{F}=
\setcounter{MaxMatrixCols}{14}
\begin{pmatrix}
 f_{11} & f_{12} & f_{13} & f_{14} & f_{15} & f_{16} & f_{21} & f_{22} & f_{23} & f_{24} &f_{25} & f_{26} \\
 w_1 &\emptyset & \emptyset & w_2 & \emptyset & \emptyset & w_3 & \emptyset & \emptyset & w_4 & \emptyset & \emptyset 
 \end{pmatrix}.
$$}\hfill $\Diamond$

The following theorem proves the output of the adapted firm-copies-proposing deferred acceptance algorithm is a stable* matching.

\begin{theorem}\label{resultado del DA es estable}
  Let  $\lambda_\mathcal{F}$ be the output of the firm-copies-proposing deferred acceptance algorithm. Then, $\lambda_\mathcal{F}$ is a stable* matching.
\end{theorem}
\begin{proof}
Assume that $\lambda_\mathcal{F}\notin S^*(M).$ Then, by Remark \ref{remark de no stable*} there are four possible situations:
\begin{enumerate}[(i)]
\item \textbf{There is $\boldsymbol{w\in \mathcal{W}$ such that $\emptyset \overline{P}_w\lambda_\mathcal{F}(w)}$}. Given that at each stage $k$,  $w$ selects $\widetilde{f}$ as the best firm-copy among those in $O^k_w\cup \lambda^{k-1}(w) $ with respect to $\overline{P}_w$, and in stage 1 $\lambda^{0}(w)=\emptyset $, we have that $\lambda_\mathcal{F}(w)\overline{R}_w \emptyset $. Then, this situation does not occur.

\item \textbf{There is $\boldsymbol{f_{ij}\in \mathcal{F}$ such that $\emptyset P_{ij}\lambda_\mathcal{F}(f_{ij})}$}. Given that at each stage $f_{ij}$ makes an offer, if it is authorized, to an acceptable worker with respect to $P_{ij}$, we have that $\lambda_\mathcal{F}(f_{ij})R_{ij}\emptyset.$ Then, this situation does not occur.

\item \textbf{There are $\boldsymbol{f_{ij},f_{ij'}\in \mathcal{F}$ such that $\lambda_\mathcal{F}(f_{ij'}) P_{ij}\lambda_\mathcal{F}(f_{ij})P_{ij}\emptyset}$}.  Assume that there are $w,w' \in \mathcal{W}$ such that $w=\lambda_\mathcal{F}(f_{ij})$ and $w'=\lambda_\mathcal{F}(f_{ij'}).$
Assume also that there are two stages $t,t'$ such that $t'$ is the stage where $f_{ij'} $ is assigned to $w'$ ($\lambda^{t'}(f_{ij'})=w'$ and $\lambda^{t'-1}(f_{ij'})\neq w'$) and such that $t$ is the stage where $f_{ij} $ is assigned to $w$ ($\lambda^{t}(f_{ij})=w$ and $\lambda^{t-1}(f_{ij})\neq w$). We have to consider two subcases.

\textbf{$\boldsymbol{t'<t}:$} Since $w'P_{ij} w$, we have that $f_{ij}$ is unauthorized to make an offer to $w$, so it is not possible that $\lambda_\mathcal{F}(f_{ij})=w.$  

\textbf{$\boldsymbol{t'\geq t}$:} Note that since $w$ is matched with $f_{ij}$ before than $w'$ and the fact that $w'P_{ij} w$ implies that $w'$ reject $f_{ij} $ in an earlier step that $t$, say $\widetilde{t}$. In this case, we claim that there is a firm-copy $f_{\widetilde{i}k}$  that is not a copy of $f_{ij}$ such that $w'=\lambda^{\widetilde{t}}(f_{\widetilde{i}k})$ and $f_{\widetilde{i}k} \overline{P}_{w'}f_{i,j}$. Since $w'$ reject $f_{ij} $ at stage $\widetilde{t}$, such a firm-copy exists that is temporarily matched to $w'$ at stage $\widetilde{t}$, and $w'$ prefers this firm-copy over $f_{ij}.$ Moreover,  since  $\widetilde{t}\leq t$, we have that $f_{\widetilde{i}k}$ is not a copy of $f_{ij}$, otherwise $f_{ij}$ will not be authorized to make an offer to $w$. Thus, in this case, the claim holds. Note that this contradicts the definition of workers' preferences $\overline{P}$ item (1), since $f_{ij'}\overline{P}_{w'} f_{\widetilde{i}k} \overline{P}_{w'}f_{i,j}.$ 

Then, by the two cases analyzed, this situation does not occur.

\item \textbf{There is a pair $\boldsymbol{ (f_{ij},w)\in \mathcal{F}\times \mathcal{W}}$ such that $\boldsymbol{f_{ij}\overline{P}_{w}\lambda_\mathcal{F}(w)}$ and $\boldsymbol{wP_{ij}\lambda_\mathcal{F}(f_{ij'})}$ for each $\boldsymbol{j'\in J_{i}}$}.
If $w P_{ij}\lambda_\mathcal{F}(f_{ij'}) $ for each $j'\in J_{i}$, then each $f_{ij'}$ has made an offer to $w$ in a previous stage and it was rejected because $w$ was temporarily matched with a better firm-copy that is not a copy of $f_{ij'}$, otherwise  $f_{ij'}$ is unauthorized to make an offer. Then, $\lambda_\mathcal{F}(w)\overline{P}_w f_{ij'}$ for each $j'\in J_i.$ Then, this situation does not occur.
 \end{enumerate}
Therefore, since none of the four cases hold, we have that $\lambda_\mathcal{F}\in S^*(M).$
\end{proof}

The algorithm presented in Table \ref{tabla algoritmo W-proposal} is the adapted worker-proposing deferred acceptance algorithm for a market $M$. This adaptation also ensures that the output is envy-free among firm-copies. We say that a worker makes a \emph{valid} offer to a firm-copy if the acceptance of the worker's proposal by that firm-copy does not generate envy among the other copies.

Let $O^k_f$ be the set of workers making a job offer to firm $f$ at stage $k$. 
\begin{table}[h!]
\centering
\begin{tabular}{l l}
\hline \hline
\multicolumn{2}{l}{\textbf{Algorithm:}}\vspace*{10 pt}\\

\textbf{Input} & A one-to-one market $M$.\\

\textbf{Output} & A matching $\lambda_{\mathcal{W}}$.\vspace*{10 pt}\\
%\textbf{Begin} & \vspace*{5 pt}\\
 & \hspace{-52pt}\texttt{DEFINE}: $\lambda^0(w)=\emptyset$ and $\lambda^0(f_{ij})=\emptyset$ for each $w\in \mathcal{W}$ and each $f_{ij}\in \mathcal{F}$.\vspace*{10 pt}\\
  \textbf{Stage $\boldsymbol{1}$} &  $\boldsymbol{(a)}$ Each worker $w$ makes an offer to the best acceptable firm-copy  at $\overline{P}_w$.\\
&  $\boldsymbol{(b)}$ Each firm-copy $f_{ij}$ selects, with respect to $P_{f_{ij}}$, the best worker\\
& \hspace{20 pt} among those in $O^1_{f_{ij}}\cup \lambda^0(f_{ij})$, say $\widetilde{w}_{f_{ij}}$.\\
& \hspace{-20 pt} Set  $R^1_{f_{ij}}=O^1_{f_{ij}} \setminus \{\widetilde{w}_{f_{ij}}\}$  for each $f_{ij}\in \mathcal{F}$. Define \\
& \hspace{40 pt}$\lambda^1(w)=\begin{cases}
\emptyset    & \text{if } w\in \cup_{f_{ij}\in \mathcal{F}} R_{f_{ij}}^1, \\ 
f_{ij} &   \text{if } w=\widetilde{w}_{f_{ij}}.\\
\end{cases}$\\
 \textbf{Stage $\boldsymbol{k}$} & $\boldsymbol{(a)}$ Each worker $w\in \bigcup_{f_{ij}\in \mathcal{F}}R^{k-1}_{f_{ij}}$ make an new offer to the best \\ 
 &\hspace{20 pt} acceptable firm-copy $f_{ij}$ with respect to $\overline{P}_w$ who has not previously rejected it.\\
& $\boldsymbol{(b)}$ For each firm-copy $f_{ij}$, define the set of ``Valid Offers'' \\
& $\widetilde{O}_{ij}^k=O_{ij}^k\setminus\{w\in O_{ij}^k: \exists w'\in \mathcal{W} \text{ with } w'=\lambda^{k-1}(f_{ij'})P_{f_{ij}} w \text{ for some }j'\in J_i \}$.\\
& Each firm $f_{ij}$ selects the best worker, with respect to $P_{f_{ij}}$\\
& \hspace{20 pt} among those in $\widetilde{O}^k_{f_{ij}}\cup \lambda^{k-1}(f_{ij})$, say $\widetilde{w}^{ij}$. \\
& \hspace{-20 pt} Set $R^k_{f_{ij}}=O^k_{f_{ij}} \cup \lambda^{k-1}(f_{ij}) \setminus \{\widetilde{w}^{ij}\}$ for each $f_{ij}\in \mathcal{F}$. Define\\
& \hspace{40 pt}$\lambda^k(w)=\begin{cases}
\emptyset    & \text{if } w\in \cup_{f_{ij}\in \mathcal{F}} R_{f_{ij}}^k, \\ 
f_{ij}   & \text{if } w=\widetilde{w}^{ij}\text{, and}\\
\lambda^{k-1}(w)& \text{otherwise.}
\end{cases}$\\
& \hspace{20 pt} \texttt{If} $\bigcup_{f_{ij}\in \mathcal{F}}R^k_{f_{ij}} =\emptyset$: \\
& \hspace{60 pt} \texttt{STOP} and let $\lambda_\mathcal{W}=\lambda^k$.\\
& \hspace{20 pt} \texttt{ELSE:} \\
& \hspace{60 pt} \texttt{CONTINUE TO STAGE $k+1$} \\
\\
\hline \hline
\end{tabular}
\caption{Adapted Worker-proposal Deferred Acceptance algorithm}
\label{tabla algoritmo W-proposal}
\end{table}

\bigskip
Before proving that the output of the algorithm is a stable* matching, we first illustrate the procedure with Example \ref{ejemplo}.
\medskip

\noindent \textbf{Example \ref{ejemplo} (continued)} Considering the associated one-to-one market previously presented, we applied the adapted worker-proposing deferred acceptance algorithm as follows:

\noindent\textit{\noindent\textbf{Stage 1:} The workers' offers are presented in the following table:}
\begin{center}
\begin{tabular}{|c|c|c|c|}\hline
$w_1$&$w_2$&$w_3$&$w_4$\\
$f_{21}$&$f_{21}$&$f_{11}$&$f_{11}$\\ \hline
\end{tabular}
\end{center} 
\textit{So, the sets of offers that firm-copies receive in stage 1 are: $O_{f_{1,1}}^{1}=\{w_3,w_4\}$ and $O_{f_{21}}^{1}=\{w_1,w_2\}.$ }
\textit{Firm-copies, according to their preferences $P$, accept the offers, and the resulting rejection sets are $R_{f_{11}}^{1} = {w_4}$ and $R_{f_{21}}^{1} = {w_2}$. Thus,} $$
\lambda^1=
\setcounter{MaxMatrixCols}{14}
\begin{pmatrix}
  f_{11} & f_{12} & f_{13} & f_{14} & f_{15} & f_{16} & f_{21} & f_{22} & f_{23} & f_{24} &f_{25} & f_{26} \\
 w_3 &\emptyset & \emptyset & \emptyset & \emptyset & \emptyset & w_2 & \emptyset & \emptyset & \emptyset & \emptyset & \emptyset 
 \end{pmatrix}.
$$

\noindent\textit{\textbf{Stage 2:} Each worker ${w_2, w_4} \in R_{f_{11}}^{1} \cup R_{f_{21}}^{1}$ submits a new offer. Worker $w_2$ offers to firm-copy $f_{22}$, while $w_4$ offers to firm-copy $f_{1,2}$. Since both offers are valid, and firm-copies $f_{12}$ and $f_{22}$ accept them without generating any further rejections, the algorithm stops. Thus, $$
\lambda_\mathcal{W}=
\setcounter{MaxMatrixCols}{14}
\begin{pmatrix}
  f_{11} & f_{12} & f_{13} & f_{14} & f_{15} & f_{16} & f_{21} & f_{22} & f_{23} & f_{24} &f_{25} & f_{26} \\
 w_3 &w_4 & \emptyset & \emptyset & \emptyset & \emptyset & w_2 & w_1 & \emptyset &\emptyset & \emptyset & \emptyset 
 \end{pmatrix}.
$$}\hfill $\Diamond$

The following theorem shows that the output of the adapted worker-proposing deferred acceptance algorithm is a stable* matching. The proof follows the same spirit as the proof of Theorem \ref{resultado del DA es estable}, adapted for the case where workers are the ones making the offers.

\begin{theorem}
   Let  $\lambda_\mathcal{W}$ be the output of the worker-proposing deferred acceptance algorithm. Then, $\lambda_\mathcal{W}$ is a stable* matching.
   \end{theorem}
\begin{proof}
Assume that $\lambda_\mathcal{W}\notin S^*(M).$ Then, by Remark \ref{remark de no stable*} there are four possible situations:
\begin{enumerate}[(i)]
\item \textbf{There is $\boldsymbol{f_{ij}\in \mathcal{F}$ such that $\emptyset P_{ij}\lambda_\mathcal{W}(f_{ij})}$}. Given that at each stage $k$,  $f_{ij}$ selects $\widetilde{w}^{ij}$ as the best worker among those in $\widetilde{O}^k_{f_{ij}}\cup \lambda^{k-1}(w) $ with respect to $P_{ij}$, and in stage 1 $\lambda^{0}(f_{ij})=\emptyset $, we have that $\lambda_ \mathcal{W}(f_{ij})R_{ij} \emptyset $. Then, this situation does not occur.

\item \textbf{There is $\boldsymbol{w\in \mathcal{W}$ such that $\emptyset \overline{P}_{w}\lambda(w)}$}. Given that at each stage $w$ makes an offer, to an acceptable firm with respect to $\overline{P}_{w}$, we have that $\lambda(w)_\mathcal{W}\overline{R}_{w}\emptyset.$ Then, this situation does not occur.

\item \textbf{There are $\boldsymbol{f_{ij},f_{ij'}\in \mathcal{F}$ such that $\lambda_\mathcal{W}(f_{ij'}) P_{ij}\lambda_\mathcal{W}(f_{ij})P_{ij}\emptyset}$}.  Assume that there are $w,w' \in \mathcal{W}$ such that $w=\lambda_\mathcal{W}(f_{ij})$ and $w'=\lambda_\mathcal{W}(f_{ij'}).$
Assume also that there are two stages $t,t'$ such that $t'$ is the stage where $f_{ij'} $ is assigned to $w'$ ($\lambda^{t'}(f_{ij'})=w'$ and $\lambda^{t'-1}(f_{ij'})\neq w'$) and such that $t$ is the stage where $f_{ij} $ is assigned to $w$ ($\lambda^{t}(f_{ij})=w$ and $\lambda^{t-1}(f_{ij})\neq w$). There are two subcases to consider.

\textbf{$\boldsymbol{t'<t}:$} Since $w'P_{ij} w$, we have that $w$ is not a valid offer for $f_{ij}$, so it is not possible that $\lambda_\mathcal{W}(f_{ij})=w.$  

\textbf{$\boldsymbol{t'\geq t}$:} Note that since $w$ is matched with $f_{ij}$ before that $w'$ and the fact that $w'P_{ij} w$ implies that $w'$ reject $f_{ij} $ in an earlier step that $t$, say $\widetilde{t}$. In this case, we claim that there is a firm $f_{\widetilde{i}k}$  that is not a copy of $f_{ij}$ such that $w'=\lambda^{\widetilde{t}}(f_{\widetilde{i}k})$ and $f_{\widetilde{i}k} \overline{P}_{w'}f_{i,j}$. Since $w'$ reject $f_{ij} $ at stage $\widetilde{t}$, such a firm-copy exists and is temporarily matched to $w'$ at stage $\widetilde{t}$, and $w'$ prefers this firm-copy over $f_{ij}.$ Moreover,  since  $\widetilde{t}\leq t$, we have that $f_{\widetilde{i}k}$ is not a copy of $f_{ij}$, otherwise $w$ is not a valid offer for $f_{ij}$. Thus, in this case, the claim holds. Note that this contradicts the definition of workers' preferences $\overline{P}$ item (1), since $f_{ij'}\overline{P}_{w'} f_{\widetilde{i}k} \overline{P}_{w'}f_{i,j}.$ 

Then, by the two cases analyzed, this situation does not occur.

\item \textbf{There is a pair $\boldsymbol{ (f_{ij},w)\in \mathcal{F}\times \mathcal{W}}$ such that $\boldsymbol{f_{ij}\overline{P}_{w}\lambda_\mathcal{W}(w)}$ and $\boldsymbol{wP_{ij}\lambda_\mathcal{W}(f_{ij'})}$ for each $\boldsymbol{j'\in J_{i}}$}.
There are two subcases to consider.

\noindent \textbf{If $\boldsymbol{w}$ makes a non-valid offer to $\boldsymbol{f_{ij}}$ at some step:} then, there is a stage $t$ and a worker $w'$ such that $w'=\lambda^{t}(f_{ij})P_{ij}w$, contradicting that $wP_{f_{ij}}\lambda_\mathcal{W}(f_{ij'})$ for each $j'\in J_i.$

\noindent \textbf{If $\boldsymbol{w}$ makes a valid offer to $\boldsymbol{f_{ij}}$ at some step:} Then, $w$ was rejected by $f_{ij}$ at some stage $k$. Thus, there is $w'\in \mathcal{W}$ such that $w'=\lambda^k(f_{ij})P_{ij}w$, contradicting that $wP_{ij}\lambda_\mathcal{W}(f_{ij'})$ for each $j'\in J_i.$
\end{enumerate}
Therefore, since none of the four cases hold, we have that  $\lambda_\mathcal{W}\in S^*(M).$
\end{proof}
   
\subsection{The Isomorphism between markets}\label{subseccion isomorfismo}
In this subsection, we prove that there is an isomorphism between the set of stable matchings in a many-to-one market where firm have path-independent choice functions, and the set of stable* matchings of an associated one-to-one market, which is derived from the Aizermann-Malishevski decomposition.

\begin{theorem}
\label{S.iso} For any Aizermann-Malishevski decomposition, there is an isomorphism between $S^*(M)$ and  $S({{\mathcal{M}}})$.
\end{theorem}
\begin{proof} Fix any Aizermann-Malishevski decomposition and, then fix a one-to-one market $M$.
Given $\lambda \in S^*\left (M\right )$, define $T(\lambda )=\mu $ as 
$$\mu (\varphi _{i})=\bigcup \limits _{j\in J_{i}}\lambda (f_{ij})$$ for each $\varphi_i\in \Phi.$ 
Moreover,  if $w\in \mu (\varphi _i)$, then $\mu (w)=\varphi _i$.

Similarly, if $\mu \in S\left ({{{\mathcal{M}}}}\right ),$ define $T^{-1}(\mu )=\lambda $ as
$\lambda(w)=f_{ij}$ for $w=\max _{\mu (\varphi_{i})}P_{ij}$ such that there is no $ j^{\prime
}<j$ with $w=\max _{\mu (\varphi_{i})}P_{ij^{\prime }}$.
Moreover, if $\lambda (w)=f_{ij}$, then $\lambda (f_{ij})=w.$

Formally we need to show that, 
\begin{enumerate}
\item[\textbf{(i)}] \label{iso.c.1}
\textbf{Given $\boldsymbol{\lambda\in S^{\ast}\left(  M\right)  $, 
then $T(\lambda)=\mu\in
S\left(  \mathcal{M}\right)}$}. We prove this implication by proving the following two claims.

\begin{claim}
$\mu=T(\lambda)$ is a many-to-one matching.
\end{claim}

By definition of $T$, for each $w\in\mathcal{W}$, $w\in \mu(\varphi_i)$ implies that there is $j\in J_i$ such that $ w=\lambda(f_{ij})$. 
Assume $\mu$ is not a matching, then there is $\varphi_{i'}$ with $ i'\neq i$ such that $w\in\mu(\varphi_{i'})$. So, there is $ j'\in J_{i'}$ such that $w=\lambda(f_{i'j'})$. 
Thus, $w=\lambda(f_{ij})=\lambda(f_{i'j'})$ which contradicts $\lambda$ being a one-to-one matching.

The construction of $T$ ensures the bilateral nature of the assignment.

\begin{claim}
$\mu=T(\lambda)\in S^{\ast}\left(  M\right)$.
\end{claim}

By way of contradiction, assume $\mu\notin S(\mathcal{M})$. Then, we have three cases to analyze:

\begin{enumerate}[(a)]

% IR for workers

\item 
If $\mu(w)=\varphi_i$, then there is  $j\in J_i$ 
such that $w=\lambda(f_{ij})$.
 Assume that $\emptyset P_w \varphi_i$, then by the construction of the one-to-one associated market (particularly condition \eqref{pref.c.1} on workers' preferences) it follows that $\emptyset \overline{P}_w f_{ij}$,  contradicting that $\lambda\in S^*(M)$.

% IR for firms

\item 
Let $w\in\mu(\varphi_i)$. 
Assume that $w\notin C_{\varphi_i}(\mu(\varphi_i))$. 
Since $\mu(\varphi_i)=\bigcup_j\lambda(f_{ij})$ 
and 
$C_\varphi(\mu(\varphi_i))=\bigcup_j \max_{\mu(\varphi)} P_{ij}$, 
then there is $j'\in J$ such that $\lambda(f_{ij'})=w$ 
and 
$\max_{\mu(\varphi_i)} P_{ij'}\neq w$. 
Then, $\lambda$ is blocked by $f_{ij'}$ 
since it satisfies condition \eqref{cond.b.ii} 
of Definition \ref{def.block*}, contradicting that $\lambda\in S^*(M)$.

% No blocking pairs

\item 
Assume that there is a blocking pair to $\mu$, i.e., there is a pair $ (\varphi_i, w)$ such that $\varphi_i P_w \mu(w) =\varphi_{i'}$ for some $i'$, and $w\in C_\varphi(\mu(\varphi_i)\cup\{w\}).$
Then, by the Aizermann-Malishevski decomposition, there is $ij'$ such that $w=\max_{\mu(\varphi_i)\cup\{w\}} P_{ij'}$ which implies that $w P_{ij'} \lambda(f_{ij''})$ for each $j''\in J_i\setminus\{j'\}$.
By the construction of the associated one-to-one market, if $\varphi_i P_w \varphi_{i'}$, then $f_{ij'}\overline{P}_w f_{i'\tilde{j}}$ for each $\tilde{j} \in J_{i'}$ (in particular for $\lambda(w)=f_{i'\hat{j}}$ with $\hat{j}\in J_{i'}$).
Hence, $(w,f_{ij'})$ is a blocking* pair to $\lambda$,  contradicting that $\lambda\in S^*(M)$.

% $w=\max\{\mu\varphi_i\cup \{w\},P_{\varphi_i}\}$

\end{enumerate}
%%%%%%%%%%%%%%%%%%%%%%%%%%%%%%%%%%%

\item[\textbf{(ii)}] \label{iso.c.2}
\textbf{Given $\boldsymbol{\mu\in S\left(  \mathcal{M}\right)  $, 
then 
$T^{-1}(\mu)=\lambda\in S^{\ast}\left(M\right)}$}. We prove this implication by proving two more claims.

\begin{claim}
$\lambda$ is a one-to-one matching. 
\end{claim}

By way of contradiction, assume  $\lambda $ is not a matching. So, there is $w\in \mathcal{W}$ such that $\lambda(w)=f_{ij}$ and $\lambda(w)=f_{i'j'}$.
If $i\neq i'$, then $w\in \mu(\varphi_i)$ and $w\in \mu(\varphi_{i'})$, which would contradict that $\mu$ is a many-to-one matching.
If $i=i'$ and $j\neq j'$, then either $j<j'$ or vice-versa, making $\lambda$ contradict the construction of $T^{-1}(\mu)$.

The construction of $T^{-1}$ ensures the bilateral nature of the assignment.

\begin{claim}
$\lambda\in S^*(M).$
\end{claim}

By way of contradiction, assume  $\lambda\notin S^*(M)$.

\begin{enumerate}[(a)]

% IR workers

\item 
Assume that $\lambda$ is blocked* by a worker $w$.
If $\lambda(w)=f_{ij}$ then $w\in \mu(\varphi_i)$. 
Since $\mu\in S(\mathcal{M})$, then $\varphi_i P_w \emptyset $. 
By construction of the preferences in the associated one-to-one market, it follows that $f_{ij} \overline{P}_w \emptyset$ for each $j\in J_i$, a contradiction.

% IR firms

\item
Assume that $\lambda$ is blocked* by a firm-copy.
Since $\mu\in S(\mathcal{M})$ then $C_{\varphi_i}(\mu(\varphi_i))=\mu(\varphi_i)=\bigcup_{j\in J_i}\max_{\mu(\varphi_i)} P_{ij}$. 
By construction of $T^{-1}$ if $\lambda(f_{ij'})=w$, then $\max_{\mu(\varphi_i)} P_{ij'}=w P_{ij'} \lambda(f_{ij''})$ for each $ j''\in J_i\setminus\{j'\}$ and $w P_{ij'} \emptyset$, a contradiction. 

\item Assume that there are $f_{ij}$ and $f_{ij'}$ such that $\lambda(f_{ij'})P_{ij}\lambda(f_{ij})P_{ij} \emptyset.$ Since $\mu\in  S(\mathcal{M})$, and $T^{-1}(\mu)=\lambda$, we have that for each $w\in \mu(\varphi_i)$ there is $P_{ij}$ such that 
$w=\max_{\mu(\varphi_i)} P_{ij}$ and $w=\lambda(f_{ij}).$ Then, for each $f_{ij}$ such that $\lambda(f_{ij})P_{ij} \emptyset$ there is no $f_{ij'}$ such that $\lambda(f_{ij'})P_{ij}\lambda(f_{ij})$, a contradiction. 
% No blocking pair
\item
Assume that $(w,f_{ij})$ form a blocking* pair, i.e., $w P_{ij}  \lambda(f_{ij'})$ for each $j'\in J_i$ 
and 
$f_{ij} \overline{P}_w \lambda(w)$.\footnote{
W.l.o.g. assume $j=\min\{j''\in J_i$ such that $w P_{ij''} \lambda(f_{ij'})$ for each $j'\in J_i\}.$ Also assume that some firm-copy is matched, otherwise stability and stability* coincide.}
By construction of $T^{-1}$ if $\lambda(f_{ij})\neq\emptyset$ then $\lambda(f_{ij})=\max_{\mu(\varphi_i)} P_{ij}$. 
It follows that $w P_{ij} \left[\max_{\mu(\varphi_i)} P_{ij}\right]$, which in turn implies that
$w=\max_{\mu(\varphi_i)\cup\{w\}} P_{ij}$. 
% By stability of $\mu$ it follows that $\mu$ is individually rational,
% $\mu(\varphi_i)=Ch(\mu(\varhpi))=\bigcup_{j'\in J_i} \max_{\mu(\varphi_i)} P_{ij'} $; 
% Where the latter equality follows from substitutability of $P_{\varphi_i}$, and the associated Aizermann-Malishevski decomposition.
Then, we have that $w\in C_{\varphi_i}(\mu(\varphi_i)\cup\{w\})=\bigcup_{j'\in J_i} \max_{\mu(\varphi_i)\cup\{w\}} P_{ij'}$. 
By construction of the associated one-to-one market, $f_{ij} \overline{P}_w f_{i'j''}$ implies $\varphi_i P_w \varphi_{i'}$. 
Then, $(w,\varphi_i)$ is a blocking pair to $\mu$, contradicting that $\mu\in S(\mathcal{M})$.
\end{enumerate}
\end{enumerate}
By (i) and (ii), $T$ is an isomorphism between the set of stable many-to-one matchings and the set of stable* one-to-one matchings.
\end{proof}

It is important to highlight that this previous result does not depend on the order of the indices used in the Aizerman-Malishevski decomposition; therefore, it holds regardless of which particular decomposition is considered. This implies that, for any such decomposition, the set of stable* matchings is isomorphic to the set of stable matchings in the original many-to-one market, and hence the various stable* sets corresponding to different one-to-one markets---obtained from different Aizerman-Malishevski decompositions---are also isomorphic to each other.

The following example illustrates the isomorphism between the many-to-one market and the related one-to-one market of Example \ref{ejemplo}.
\medskip

\noindent \textbf{Example \ref{ejemplo} (continued)} \textit{ Recall that the stable matchings of the many-to-one subsitutable market are:}
\begin{align*}
    \mu_\Phi=
\setcounter{MaxMatrixCols}{14}
\begin{pmatrix}
 \varphi_1 & \varphi_2  \\
 w_1w_2 &w_3w_4
 \end{pmatrix}
     \mu_1=
\setcounter{MaxMatrixCols}{14}
\begin{pmatrix}
 \varphi_1 & \varphi_2  \\
 w_1w_3 &w_2w_4
 \end{pmatrix}\\
     \mu_2=
\setcounter{MaxMatrixCols}{14}
\begin{pmatrix}
 \varphi_1 & \varphi_2  \\
 w_2w_4 &w_1w_3
 \end{pmatrix}
     \mu_\mathcal{W}=
\setcounter{MaxMatrixCols}{14}
\begin{pmatrix}
 \varphi_1 & \varphi_2  \\
 w_3w_4 &w_1w_2
 \end{pmatrix}
\end{align*}

\noindent\textit{The stable* matchings of the associated one-to-one market are:} 
$$
\lambda_\mathcal{F}=
\setcounter{MaxMatrixCols}{14}
\begin{pmatrix}
 f_{11} & f_{12} & f_{13} & f_{14} & f_{15} & f_{16} & f_{21} & f_{22} & f_{23} & f_{24} &f_{25} & f_{26} \\
 w_1 &\emptyset & \emptyset & w_2 & \emptyset & \emptyset & w_3 & \emptyset & \emptyset & w_4 & \emptyset & \emptyset 
 \end{pmatrix}$$
$$ \lambda_1=
\setcounter{MaxMatrixCols}{14}
\begin{pmatrix}
  f_{11} & f_{12} & f_{13} & f_{14} & f_{15} & f_{16} & f_{21} & f_{22} & f_{23} & f_{24} &f_{25} & f_{26} \\
 w_2 &\emptyset & w_4 &\emptyset & \emptyset & \emptyset & w_3 & \emptyset & \emptyset & \emptyset & \emptyset & w_1 
 \end{pmatrix}$$
$$ \lambda_2=
\setcounter{MaxMatrixCols}{14}
\begin{pmatrix}
 f_{11} & f_{12} & f_{13} & f_{14} & f_{15} & f_{16} & f_{21} & f_{22} & f_{23} & f_{24} &f_{25} & f_{26} \\
 w_1 &\emptyset & \emptyset &\emptyset & \emptyset &  w_3 & w_4& \emptyset & w_2  & \emptyset & \emptyset & \emptyset 
 \end{pmatrix}$$

 $$   \lambda_\mathcal{W}=
\setcounter{MaxMatrixCols}{14}
\begin{pmatrix}
 f_{11} & f_{12} & f_{13} & f_{14} & f_{15} & f_{16} & f_{21} & f_{22} & f_{23} & f_{24} &f_{25} & f_{26} \\
 w_3 &w_4 & \emptyset & \emptyset & \emptyset & \emptyset & w_2 & w_1 & \emptyset &\emptyset & \emptyset & \emptyset 
 \end{pmatrix}$$
   
\noindent \textit{It easy to see that $T(\lambda_\mathcal{F})=\mu_\Phi$, $T(\lambda_1)=\mu_1$, $T(\lambda_2)=\mu_2$, and $T(\lambda_\mathcal{W})=\mu_\mathcal{W}.$}\hfill $\Diamond$\medskip

One of the most significant results in the matching literature is the well-known \emph{Rural Hospital Theorem}. In many-to-one markets, this result requires not only that firms’ are endowed with path-independent choice functions, but also that they obey the \emph{Law of Aggregate Demand (LAD)}. The LAD states that when a firm chooses from a larger set of workers, it selects at least as many as it did from any of its subsets. Formally, given $\varphi \in \Phi$, if $W'' \subseteq W' \subseteq W$, then $|C_\varphi(W'')| \leq |C_\varphi(W')|$.\footnote{This property was first studied by \citet{Alkan2002} under the name \textit{cardinal monotonicity}. See also \citet{Hatfield2005}.}

The version of the Rural Hospital Theorem for many-to-one markets with path-independent choice functions satisfying the LAD states that \emph{each agent is matched with the same number of partners in every stable matching.}\footnote{The Rural Hospital Theorem has been established in various contexts by multiple authors \citep[see][among others]{McVitie1971, Roth1984a, Roth1985a, Martinez2000}. The version that applies to our setting is the one presented by \citet{Alkan2002} for many-to-many markets.}

The following theorem provides an adaptation of this classical result to the associated one-to-one market. Its proof relies fundamentally on the existence of the isomorphism with the many-to-one setting.
\begin{theorem}[The adapted Rural Hospital Theorem]
    The number of matched firm-copies of a firm remains the same in every stable* matching.
\end{theorem}
\begin{proof}
  By way of contradiction, assume that there are two stable* matchings $\lambda$ and $\lambda'$ such that the number of matched firm-copies differs. That is, there is a firm $\varphi_i \in \Phi$ for which
\begin{equation}\label{ecu 1 para THR} |f_{ij} \in \mathcal{F}_i : \lambda(f_{ij}) \neq \emptyset | \neq |f_{ij} \in \mathcal{F}_i : \lambda'(f_{ij}) \neq \emptyset|. \end{equation}
where $\mathcal{F}_i=\{f_{ij}\in \mathcal{F}:j\in J_i\}$, i.e., the set of firm-copies of firm $\varphi_i$. Now, by Theorem~\ref{S.iso}, there  are stable matchings $\mu = T(\lambda)$ and $\mu' = T(\lambda')$ in the many-to-one market. Then, by \eqref{ecu 1 para THR}, it follows that $|\mu(\varphi_i)| \neq |\mu'(\varphi_i)|$, contradicting the many-to-one version of the Rural Hospital Theorem.
\end{proof}

It is worth highlighting that this result is also independent of a particular Aizerman-Malishevski decomposition. Therefore, it applies not only to the set of stable* matchings of a given associated one-to-one market, but also uniformly across all possible associated one-to-one markets arising from different decompositions. Formally, 
\begin{corollary}
    For any firm, the number of its matched copies is invariant across all stable* matchings in every associated one-to-one market.
\end{corollary}

\section{Final remarks}\label{seccion final remaks}

In this paper, we present a method to decompose a many-to-one market where firms have path-independent choice functions into a one-to-one market using the Aizerman-Malishevski decomposition. We define a notion of stability*  for an associated one-to-one market that is adjusted to show that the set of stable matchings in a many-to-one market and the set of stable* matchings in an associated one-to-one market are isomorphic. Furthermore, we present an adaptation of the deferred acceptance algorithm  \citep[originally introduced by][]{Gale1962} for our associated one-to-one market. Regardless of which side proposes --be it the firm-copies or the workers-- a stable* matching is always obtained. Although the fact that both markets are isomorphic already indicates that the set of stable* matchings in the related one-to-one market is non-empty, the algorithm provides an alternative proof of this fact.

One of the central results in matching theory is the Rural Hospital Theorem, which in many-to-one markets requires path-independent choice functions satisfying the Law of Aggregate Demand (LAD). We adapt this result to our one-to-one setting, showing that the number of matched firm-copies remains invariant across all stable* matchings. This invariance holds independently of the particular Aizerman-Malishevski decomposition used, applying uniformly across all associated one-to-one markets.

A common result when using the deferred acceptance algorithm across all matching markets is that the output is proposing-side optimal: if, for instance, the firms are the ones making the offers, the output of the algorithm is the firm-optimal stable matching. Unfortunately, this is not valid in our associated one-to-one market.
 Recall from Example \ref{ejemplo} the stable* matchings resulting from the adapted deferred acceptance algorithm when both the firm-copies and the workers propose:

 $$
\lambda_\mathcal{F}=
\setcounter{MaxMatrixCols}{14}
\begin{pmatrix}
  f_{11} & f_{12} & f_{13} & f_{14} & f_{15} & f_{16} & f_{21} & f_{22} & f_{23} & f_{24} &f_{25} & f_{26} \\
 w_1 &\emptyset & \emptyset & w_2 & \emptyset & \emptyset & w_3 & \emptyset & \emptyset & w_4 & \emptyset & \emptyset 
 \end{pmatrix}$$
 and 
 $$   \lambda_\mathcal{W}=
\setcounter{MaxMatrixCols}{14}
\begin{pmatrix}
 f_{11} & f_{12} & f_{13} & f_{14} & f_{15} & f_{16} & f_{21} & f_{22} & f_{23} & f_{24} &f_{25} & f_{26} \\
 w_3 &w_4 & \emptyset & \emptyset & \emptyset & \emptyset & w_2 & w_1 & \emptyset &\emptyset & \emptyset & \emptyset 
 \end{pmatrix},$$
 respectively. If we observe, for instance, the firm \(f_{12}\), we find that 
 \begin{equation}\label{ecuacion en final remarks}
     w_4 = \lambda_{\mathcal{W}}(f_{12}) P_{12} \lambda_{\mathcal{F}}(f_{12})= \emptyset,
 \end{equation} indicating that the optimality of $\lambda_{\mathcal{F}}$ fails.

 Let us consider the following result on isomorphic lattices: \emph{Given two partially ordered sets $(A, >_A)$ and $(B, >_B)$ and an isomorphism $\mathcal{I}$ between the sets $A$ and $B$ that preserves the order: for $a, b \in A$ such that $a >_A b$, it follows that $\mathcal{I}(a) >_B \mathcal{I}(b)$, if $(A, >_A)$ has a lattice structure, then $(B, >_B)$ also has a lattice structure}  \citep[see][for a thorough treatment of lattice theory]{Birkhoff1967}.

If we consider the orders $P_\mathcal{W}$ and $\overline{P}_\mathcal{W}$ induced by workers' preferences in the many-to-one and related one-to-one markets, we can assert that the set of stable* matchings has a lattice structure with respect to the induced order $\overline{P}_\mathcal{W}$ since the ordering between matchings is preserved. However, the previous example also shows that the same cannot be stated for the firm-copies side. Despite that $\mu_\Phi$ is unanimously preferred (either considering partial order induced by firms' preferences or considering Blair's partial order \citep{Blair1988}) by all firms to  $\mu_\mathcal{W}$ in the many-to-one market, $\mu_\Phi = T(\lambda_\mathcal{F})$, and $ \mu_\mathcal{W}= T(\lambda_\mathcal{W})$, \eqref{ecuacion en final remarks} shows that the isomorphisms $T$ do not preserve the order. An aspect to consider in future work is the definition of an appropriate ordering over the set of stable* matchings based on firm-copies preferences, which could allow for proving that this set admits a lattice structure.

\end{document}